\documentclass{sig-alternate}

\usepackage{mathrsfs}
\usepackage{algpseudocode}
\usepackage{algorithm}

\newdef{definition}{Definition}
\newdef{example}{Example}
\newtheorem{theorem}{Theorem}
\newtheorem{lemma}{Lemma}
\newtheorem{corollary}{Corollary}
\newtheorem{fact}{Fact}

\DeclareMathOperator{\pred}{pred}
\DeclareMathOperator*{\myand}{\odot}
\DeclareMathOperator*{\myor}{+}
\DeclareMathOperator*{\mfp}{MFP}
\DeclareMathOperator*{\mop}{MOP}
\DeclareMathOperator*{\level}{level}

\DeclareMathOperator*{\len}{len}

\DeclareMathOperator*{\bigO}{O}

\def\sharedaffiliation{%
\end{tabular}
\begin{tabular}{c}}
\makeatletter
\let\@copyrightspace\relax
\makeatother

\begin{document}
%

\title{On the computational complexity of Data Flow Analysis}
%
%
%
%
%

\numberofauthors{2} 
%
\author{
%
%
\alignauthor
Gaurav Sood\\
       \email{gauravsood0289@gmail.com}
\alignauthor
K. Murali Krishnan\\
       \email{kmurali@nitc.ac.in}       
\sharedaffiliation
       \affaddr{Department of Computer Science and Engineering,}\\
       \affaddr{National Institute of Technology Calicut,}\\
       \affaddr{Calicut - 673601, Kerala, India}\\
}

\maketitle
\begin{abstract}
We consider the problem of Data Flow Analysis over monotone data flow frameworks with a finite lattice. The problem of computing the Maximum Fixed Point (MFP) solution is shown to be $\mathscr{P}$\textit{-complete} even when the lattice has just four elements. This shows that the problem is unlikely to be efficiently parallelizable. It is also shown that the problem of computing the Meet Over all Paths (MOP) solution is $\mathscr{NL}$\textit{-complete} (and hence efficiently parallelizable) when the lattice is finite even for non-monotone data flow frameworks. These results appear in contrast with the fact that when the lattice is not finite, solving the MOP problem is undecidable and hence significantly harder than the MFP problem which is polynomial time computable for lattices of finite height.
\end{abstract}

\category{F.1.3}{Computation by Abstract Devices}{Complexity Measures and Classes}[Reducibility and completeness]
\category{D.3.4}{Programming Languages}{Processors}[Optimization]
\category{F.3.2}{Logics and Meanings of Programs}{Semantics of Programming Languages}[Program analysis]

\terms{Theory}

\keywords{Maximum fixed point solution, Meet over all paths solution, $\mathscr{P}$\textit{-complete}, $\mathscr{NL}$\textit{-complete}
} 

\section{Introduction}
The problem of data flow analysis over a monotone data framework with a bounded meet semilattice has been well studied in the context of static program analysis and machine independent compiler optimizations \cite[Sec.~9.2, 9.3]{Aho2006Compilers}. Although the meet over all paths (MOP) solution is desirable, its computation is undecidable in general \cite{springerlink:10.1007/BF00290339}. Iterative fixed point methods as in \cite{Kildall:1973:UAG:512927.512945} are commonly used to find the maximum fixed point (MFP) solution as a conservative approximation to the MOP solution \cite{springerlink:10.1007/BF00290339}. Several important problems like reaching definitions analysis, live variable analysis and available expressions analysis \cite[Sec.~9.2.4--9.2.6]{Aho2006Compilers} are essentially data flow analysis problems over monotone data flow frameworks with a bounded meet semilattice.

In this paper, the computational complexity of MFP and MOP data flow analysis problems over monotone data flow frameworks with a finite bounded meet semilattice is investigated. Since a finite bounded meet semilattice is essentially a finite lattice, we define the problem over finite lattices. We show that computing the MFP solution to data flow analysis problem over a monotone data flow framework with a finite lattice is $\mathscr{P}$\textit{-complete}. We further show that the problem of finding MOP solution is non-deterministic log space complete ($\mathscr{NL}$\textit{-complete}). In fact the proof in Section \ref{Section_MOP} shows that MOP is $\mathscr{NL}$\textit{-complete} even if the functions associated with the vertices of the control flow graph are non-monotone. These results indicate that the MFP problem is unlikely to be in the complexity class $\mathscr{NC}$ (and hence fast parallel algorithms are unlikely to exist for the problem \cite[Chap.~5]{GreHooRuz1995}). The $\mathscr{NL}$ complexity of MOP problem yields $\bigO(\log^2n)$ depth, polynomial sized parallel circuit for the problem. This further leads to the observation that although MOP computation is harder than MFP computation in general, computing MOP solution appears significantly easier than computing MFP over finite lattices.
\section{Background}
Let $(L, \leqslant)$ be a partially ordered set. Let $\bigvee$ and $\bigwedge$ respectively denote the join and meet operations in $L$. A partially ordered set $(L, \leqslant)$ is a meet semilattice, denoted by $(L, \bigwedge)$, if $x \bigwedge y$ exists for all $x, y \in L$. A meet semilattice $(L, \bigwedge)$ is a lattice, denoted by $(L, \bigvee, \bigwedge)$, if $x \bigvee y$ exists for all $x, y \in L$. A meet semilattice $(L, \bigwedge)$ is a bounded meet semilattice, denoted by $(L, \bigwedge, 1)$, if there exists an element $1 \in L$ such that $l \bigwedge 1 = l$ for all $l \in L$. A lattice $(L, \bigvee, \bigwedge)$ is a bounded lattice, denoted by $(L, \bigvee, \bigwedge, 0, 1)$, if there exist elements $0, 1 \in L$ such that $l \bigvee 0 = l$ and $l \bigwedge 1 = l$ for all $l \in L$. A finite bounded meet semilattice $(L, \bigwedge, 1)$ is essentially a finite lattice where join operation is defined as follows: For all $l, l' \in L$, $l \bigvee l' = \bigwedge \left\{ l'' \in L \mid l \leqslant l'' \text{ and } l' \leqslant l'' \right\}$. A 
lattice is complete lattice if $\bigvee S$ and $\bigwedge S$ exist for all $S \subseteq L$. It is easy to see that a finite lattice is complete. \cite{Davey2002Introduction}

Let $L = \left\{ l_1, l_2, \ldots l_m \right\}$ be a finite lattice. Let $L^n = \left\{ (l_{i_1}, \ldots l_{i_n}) \mid \forall 1 \leqslant j \leqslant n, l_{i_j} \in L \right\}$. The tuple $(l_{i_1}, l_{i_2}, \ldots l_{i_n}) \in L^n$ will be denoted by $\langle l_i \rangle_n$ or simply by $\langle l_i \rangle$ when there is no ambiguity about the index set. Let $\mathit{\ell_j}$ denote the $j^{th}$ element of $\langle l_i \rangle$. For all $\langle l_i \rangle, \langle l'_i \rangle \in L^n$, $\langle l_i \rangle \bigvee \langle l'_i \rangle = \langle l_i \bigvee l'_i \rangle$ and $\langle l_i \rangle \bigwedge \langle l'_i \rangle = \langle l_i \bigwedge l'_i \rangle$. \cite{Davey2002Introduction}

A function $f \colon A \to B$ is monotone if for all $x, y \in A$, $x \leqslant y$ implies $f(x) \leqslant f(y)$.

Let $G = (V, E)$ be a directed graph. Let $\deg^-(v)$ and $\deg^+(v)$ respectively denote the indegree and outdegree of vertex $v$. The function $\pred \colon V \to 2^V$ is defined as follows
\begin{equation*}
 \pred(v) = \left\{ u \mid (u, v) \ in \ E \right\}
\end{equation*}
\begin{definition}
 Let $L$ be a lattice and let $f \colon L \to L$. An element $l \in L$ is called a \emph{fixed point} of $f$ if $f(l) = l$. An element $l \in L$ is called the \emph{maximum fixed point} of $f$ if it is a fixed point of $f$ and for every $l' \in L$ whenever $f(l') = l'$ then $l' \leqslant l$. Let $\mfp(f)$ denote the maximum fixed point of $f$ whenever it exists.
\end{definition}
\subsection{Data Flow Analysis (DFA)}
\begin{definition}
 A \emph{control flow graph} is a finite directed graph $G = (V, E, v_s, v_t)$ where $V = \left\{ v_1, v_2, \ldots v_n \right\}$, $v_s \in V$, called \textit{entry}, is a unique vertex satisfying $\deg^-(v_s) = 0$ and $v_t \in V$, called \textit{exit}, is a unique vertex satisfying $\deg^+(v_t) = 0$. Every vertex $v_i \in V$ is reachable from $v_s$.
\end{definition}
\begin{definition}
 A \emph{monotone data flow framework} \cite{springerlink:10.1007/BF00290339} is a triple $D = (L, \bigwedge, F)$ where
 \begin{itemize}
  \item $(L, \bigvee, \bigwedge, 0, 1)$ is a finite lattice where $L = \left\{ l_1, l_2, \ldots l_m \right\}$ with $l_1 = 0$ and $l_m = 1$;
  \item $\bigwedge$ is the confluence operator; and
  \item $F$ is a collection of monotone functions from $L$ to $L$.
  \end{itemize}
\end{definition}
\begin{definition}
 A \emph{Data Flow Analysis (DFA) system} is a 5-tuple $\alpha = (G, D, M, v_\theta, l_\phi)$ where
 \begin{itemize}
  \item $G$ is a control flow graph;
  \item $D$ is a monotone data flow framework;
  \item $M \colon V \to F$ assigns a function $f_i \in F$ to the vertex $v_i$ of $G$; and
  \item $v_\theta \in V$, $l_\phi \in L$.
 \end{itemize}
 \end{definition}
 \subsection{Maximum Fixed Point (MFP) problem} \label{MFP_definition}
 Let $(G, D, M, v_\theta, l_\phi)$ be a DFA system where $G = (V, E, v_s, v_t)$ and $D = (L, \bigwedge, F)$. Then $\overline{f} \colon L^n \to L^n$ is defined as follows
 \begin{equation}
  \overline{f} \Big( \big \langle l_i \big \rangle \Big) = \Big\langle f_{i} \big(\bigwedge_{v_j \in \pred(v_i)} \ell_j \big) \Big\rangle \label{overline_f_def}
 \end{equation}
 Here we assume that $\bigwedge S = 1$ when $S = \emptyset$. With this convention, it is easy to see that $\overline{f}$ is well defined and monotone in $L^n$.
 \begin{theorem}[Knaster--Tarski theorem \cite{tarski1955lattice}]
  Let $(L, \bigvee, \bigwedge, 0, 1)$ be a complete lattice and let $f \colon L \to L$ be a monotone function. Then the MFP of $f$ exists and is unique. \label{knaster_tarski_thrm}
 \end{theorem}
 Since every finite lattice is complete, it is clear from Theorem \ref{knaster_tarski_thrm} that the MFP of $\overline{f}$ exists and is unique. Suppose $\langle l_i \rangle$ is the MFP of $\overline{f}$, then we use the notation $\mfp(v_j)$ for the element $\ell_j$.
\begin{definition}
 \emph{Maximum Fixed Point (MFP) problem:} Given a DFA system $(G, D, M, v_\theta, l_\phi)$, decide whether $\mfp(v_\theta) = l_\phi$.
\end{definition}
\subsection{Meet Over all Paths (MOP) problem}
 A path $p$ from vertex $v_{i_1}$ to vertex $v_{i_k}$ in a graph $G$, called a $v_{i_1}$-$v_{i_k}$ path, is a non-empty alternating sequence $v_{i_1}e_{i_1}v_{i_2}e_{i_2} \ldots e_{i_{k - 1}}v_{i_k}$ of vertices and edges such that $e_{i_j} = (v_{i_j}, v_{i_{j + 1}})$ for all $1 \leqslant j \leqslant k - 1$. A $v_{i_1}$-$v_{i_k}$ path is written simply as $v_{i_1}v_{i_2} \ldots v_{i_k}$ when the edges in question are clear. It may be noted that vertices and edges on a path may not be distinct. The length of path $p$ is denoted by $\len(p)$. Let $f_p = f_{i_k} \circ \cdots \circ f_{i_1}$ be called the path function associated with path $p$. Let $P_{ij}$ be the set of all paths from vertex $v_i$ to vertex $v_j$ in $G$.
\begin{definition}
 Given a DFA system $\alpha = (G, D, M, v_\theta, l_\phi)$, the \emph{meet over all paths} solution, denoted MOP, is defined as follows
  \begin{equation*}
  \mop(v_i) = \bigwedge_{p \in P_{si}} f_p(1)
  \end{equation*}
\end{definition}
Since $L^n$ is finite and hence complete, though there could be infinitely many $v_s$-$v_i$ paths, $\mop(v_i)$ is well defined by taking the infimum of all path functions.
 \begin{definition}
 \emph{Meet Over all Paths (MOP) problem:} Given a DFA system $(G, D, M, v_\theta, l_\phi)$, decide whether $\mop(v_\theta) = l_\phi$.
\end{definition}
\subsection{Monotone Circuit Value (MCV) problem}
This problem is used for reduction in Section \ref{Section_MFP} to prove that MFP is $\mathscr{P}$\textit{-complete}.
\begin{definition}
 A \emph{monotone Boolean circuit} \cite[p.~27, 122]{GreHooRuz1995} is a 4-tuple $C = (G, I, v_\theta, \tau)$ where
 \begin{itemize}
  \item $G = (V, E)$ is a finite directed acyclic graph where $V = \left\{ v_1, v_2, \ldots v_n \right\}$, and for all $v_i \in V, \deg^-(v_i) \in \left\{ 0, 2 \right\}$;
  \item $I = \left\{ v_i \in V \mid \deg^-(v_i) = 0 \right\}$ is the set of input vertices;
  \item  $v_\theta \in V$, called \textit{output}, is the unique vertex in $G$ satisfying $\deg^+(v_\theta) = 0$; and
  \item $\tau \colon V \to \left\{ \myand, \myor \right\}$ assigns either the Boolean AND function (denoted by $\odot$) or the Boolean OR function (denoted by $+$) to each vertex of $G$.
 \end{itemize}
\end{definition}
 Let $u_j$ be the $j^{th}$ input vertex of a Boolean circuit $C$ and let $\langle x_i \rangle_{\lvert I \rvert} \in \left\{ 0, 1 \right\}^{\lvert I \rvert} $ be the input to the circuit. The \emph{input value assignment} is a function $\nu \colon I \to \left\{ 0, 1 \right\}$ defined as follows
 \begin{equation*}
  \nu(u_j) = x_j \qquad \forall u_j \in I, x_j \in \left\{ 0, 1 \right\}
 \end{equation*}
 The function $\nu \colon I \to \left\{ 0, 1 \right\}$ can be extended to the function $\nu \colon V \to \left\{ 0, 1 \right\}$ called \emph{value of a node} defined as follows
  \begin{equation*}
  \nu(v_k) = \left\{
  \begin{array}{cl}
   \nu(v_i) \myand \nu(v_j) & \text{if } \tau(v_k) = \myand \text{ and } \\ & \pred(v_k) = \left\{ v_i, v_j \right\},\\
   \nu(v_i) \myor \nu(v_j) & \text{if } \tau(v_k) = \myor \text{ and } \\ & \pred(v_k) = \left\{ v_i, v_j \right\}
  \end{array} \right.
 \end{equation*}
 It is easy to see that $\nu$ is well defined when $G$ is a directed acyclic graph.
\begin{definition}
 An \emph{instance} of Monotone Circuit Value (MCV) problem is a pair $(C, \nu)$ with $C = (G, I, v_\theta, \tau)$ where
 \begin{itemize}
  \item $C$ is a monotone Boolean circuit; and
  \item $\nu \colon I \to \left\{ 0, 1\right\}$ is an input value assignment.
\end{itemize}
\end{definition}
\begin{definition}
 \emph{Monotone Circuit Value (MCV) problem:} Given an instance $(C, \nu)$ of MCV, decide whether $\nu(v_\theta) = 1$ \cite[p.~122]{GreHooRuz1995}.
\end{definition}
\subsection{Graph Meet Reachability (GMR) problem}
This problem will be used as an intermediate problem in Section \ref{Section_MOP} for showing that MOP is $\mathscr{NL}$\textit{-complete}.
\begin{definition}
 Let $A = \left\{ a_1, a_2, \ldots a_n \right\}$ be a finite set and $(L, \bigvee, \bigwedge, 0, 1)$ be a finite lattice where $L = \left\{ l_1, l_2, \ldots l_m \right\}$. A directed graph $G = (V, E)$ is said to be a \emph{product graph} of $A$ and $L$ if
 \begin{itemize}
  \item $V = \left\{ v_{ij} \mid a_i \in A, l_j \in L \right\}$ is the set of vertices; and
  \item $E \subseteq V \times V$ is the set of directed edges.
 \end{itemize}
\end{definition}
\begin{definition}
 An \emph{instance} of Graph Meet Reachability (GMR) problem is a 6-tuple $(G, A, L, v_{\theta\phi}, a_{\theta'}, l_{\phi'})$ where
 \begin{itemize}
  \item $G = (V, E)$ is a product graph of $A$ and $L$;
  \item $v_{\theta\phi} \in V$ where $a_\theta \in A$ and $l_\phi \in L$;
  \item $a_{\theta'} \in A$; and
  \item $l_{\phi'} \in L$.
 \end{itemize}
\end{definition}
Let $R_{i} = \left\{ l_j \mid v_{ij} \text{ is reachable from } v_{\theta\phi} \right\}$
\begin{definition}
\emph{Graph Meet Reachability (GMR) problem:} Given an instance $(G, A, L, v_{\theta\phi}, a_{\theta'}, l_{\phi'})$ of GMR, decide whether
\begin{equation*}
 \bigwedge_{l_i \in R_{\theta'} } l_i = l_{\phi'}
\end{equation*}
\end{definition}
\subsection{Graph Reachability (GR) problem}
Graph Reachability problem is a well known $\mathscr{NL}$\textit{-complete} problem which will be used for reduction in this paper.
\begin{definition}
 An \emph{instance} of Graph Reachability (GR) problem is a triple $(G, v_s, v_t)$ where
  \begin{itemize}
   \item $G = (V, E)$ is a directed graph; and
   \item $v_s, v_t \in V$
  \end{itemize}
\end{definition}
\begin{definition}
 \emph{Graph Reachability (GR) problem:} Given an instance $(G, v_s, v_t)$ of GR, decide whether $v_t$ is reachable from $v_s$.
\end{definition}
\begin{fact}
 GR is $\mathscr{NL}$\textit{-complete} \cite[Theorem 16.2 on p.~398]{Papadimitriou93}.
\end{fact}
\section{Related work}
It is shown in \cite{springerlink:10.1007/BF03036473} that the problem of finding \emph{meet over all valid paths (MVP) solution} to the interprocedural data flow analysis over a distributive data flow framework with possibly infinite (resp.~finite subset) semilattice  is $\mathscr{P}$\textit{-hard} (resp.~$\mathscr{P}$\textit{-complete}).

It is shown in \cite{springerlink:10.1007/BF03036473, Reps:1995:PID:199448.199462} that the problem of finding MFP and MOP solution to data flow analysis over a distributive data flow framework with a distributive sublattice of the power set lattice of a finite set is reducible to \emph{graph reachability} problem. Hence the problem is non-deterministic log space computable i.e., belongs to the complexity class $\mathscr{NL}$ (see \cite[p.~142]{Papadimitriou93} for definition). Since $\mathscr{NL} \subseteq \mathscr{NC}$ \cite[Theorem 16.1 on p.~395]{Papadimitriou93}, \cite{borodin1977relating} and $\mathscr{NC}$ admits fast parallel solutions, these results show that the above problem admits fast parallel algorithms.\\

The following is an outline for rest of the paper. In Section \ref{Section_MFP}, we show that MFP is $\mathscr{P}$\textit{-complete} by reduction from MCV. In Section \ref{Section_GMR}, we give  an $\mathscr{NL}$ algorithm for computing GMR. In Section \ref{Section_MOP}, we prove that MOP is log space reducible to GMR thereby showing that MOP is in $\mathscr{NL}$. Completeness of MOP w.r.t.~the class $\mathscr{NL}$ follows easily by a log space reduction from GR to MOP.
\section{MFP is P-complete} \label{Section_MFP}
In this section, we give a log space reduction from MCV to MFP. Since MFP is in $\mathscr{P}$ \cite{Kildall:1973:UAG:512927.512945} and MCV is $\mathscr{P}$\textit{-complete} \cite{Goldschlager:1977:MPC:1008354.1008356}, the reduction implies that MFP is also $\mathscr{P}$\textit{-complete}.
\subsection{A reduction from MCV to MFP} \label{Subsec_MFP_redct}
Given an instance $\alpha = (C, \nu)$ of MCV with $C = (G, I, v_\theta, \tau)$. Construct an instance of MFP $\alpha' = (G', D, M, v^1_\theta, (1, 1))$ as follows
\begin{itemize}
 \item $G' = (V', E', v^1_0, v^1_\theta)$ where
 \begin{itemize}
  \item $V' = \left\{ v^1_0 \right\} \bigcup \left\{ v_i^0 \mid v_i \in V(G) \setminus I \right\}\\ \text{ } \qquad \bigcup \left\{ v_i^1 \mid v_i \in V(G) \right\}$
  \item $E'$ is defined as follows
  \begin{itemize}
   \item For each input vertex $v_i \in I$ add the edge $(v^1_0, v^1_i)$ to $E'$
   \item for each vertex $v_i \in V$ add the edge $(v^0_i, v^1_i)$ to $E'$
   \item for each vertex $v_k \in V \setminus I$ with predecessors $v_i$ and $v_j$ with $i < j$ add the edges $(v^1_i, v^0_k)$, $(v^1_j, v^1_k)$ to $E'$. Note that each $v^0_k \in V \setminus I$ has a unique predecessor in $G$.
  \end{itemize}
 \end{itemize}
 \item $D = (L, \bigwedge, F)$ is defined as follows
  \begin{itemize}
   \item $L = \left\{ (0, 0), (0, 1), (1, 0), (1, 1) \right\}$ where
   \begin{itemize}
    \item $\bigvee = $ bitwise $\myor$ operation in $\left\{ 0, 1 \right\} \times \left\{ 0, 1 \right\}$
    \item $\bigwedge = $ bitwise $\myand$ operation in $\left\{ 0, 1 \right\} \times \left\{ 0, 1 \right\}$
   \end{itemize}
   \item $\bigwedge$ is the confluence operator
   \item $F = \left\{ g_I, g_0, g_1, g_{sw}, g_{\myand}, g_{\myor} \right\}$ where
   \begin{itemize}
    \item $g_I \colon L \to L$ is the identity function
    \item $g_0 \colon L \to L$ is defined as follows
    \begin{equation*}
     g_0((a_1, a_2)) = (1, 0) \qquad \forall (a_1, a_2) \in L
    \end{equation*}
    \item $g_1 \colon L \to L$ is defined as follows
    \begin{equation*}
     g_1((a_1, a_2)) = (1, 1) \qquad \forall (a_1, a_2) \in L
    \end{equation*}
    \item The swap function $g_{sw} \colon L \to L$ is defined as follows
    \begin{equation*}
     g_{sw}((a_1, a_2)) = (a_2, a_1) \qquad \forall (a_1, a_2) \in L
    \end{equation*}
    \item $g_{\myand} \colon L \to L$ is defined as follows
    \begin{equation*}
     g_{\myand}((a_1, a_2)) = (1, a_1 \myand a_2) \qquad \forall (a_1, a_2) \in L
    \end{equation*}
    \item $g_{\myor} \colon L \to L$ is defined as follows
    \begin{equation*}
     g_{\myor}((a_1, a_2)) = (1, a_1 \myor a_2) \qquad \forall (a_1, a_2) \in L
    \end{equation*}
   \end{itemize}
   It is easy to see that all functions in $F$ are monotone.
 \end{itemize}
    \item $M \colon V' \to F$ is defined as follows
      \begin{equation*}
       M(v^1_0) = f^1_0 = g_I
      \end{equation*}
      \begin{equation*}
       M(v^0_i) = f^0_i = g_{sw} \qquad \forall v^0_i \in V'
      \end{equation*}
      \begin{equation*}
       M(v^1_i) = f^1_i = \left\{
	\begin{array}{cl}
	 g_0 & \text{ if } v_i \in I \text{ and } \nu(v_i) = 0\\
	 g_1 & \text{ if } v_i \in I \text{ and } \nu(v_i) = 1
	\end{array}\right.
      \end{equation*}
      More compactly $f^1_i = g_{\nu(v_i)}$ if $i \in I$.
      \begin{equation*}
       M(v^1_i) = f^1_i = \left\{
	\begin{array}{cl}
	 g_{\myand} & \text{ if } v_i \in V \setminus I \text{ and } \tau(v_i) = \myand\\
	 g_{\myor} & \text{ if } v_i \in V \setminus I \text{ and } \tau(v_i) = \myor
	\end{array}\right.
      \end{equation*}
\end{itemize}
\begin{figure}
\centering
\includegraphics{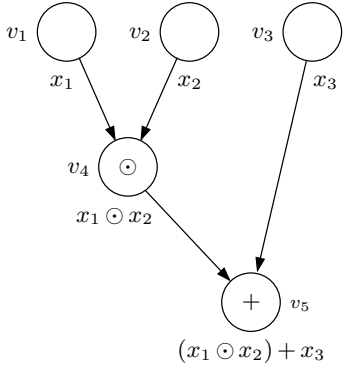}
\caption{An instance of MCV}
\label{fig_MCVP}
\end{figure}
Figure \ref{fig_MCVP} shows an instance of MCV where $I = \left\{ v_1, v_2, v_3 \right\}$. Figure \ref{fig_MFPP} shows an instance of MFP constructed from the MCV instance of Figure \ref{fig_MCVP}.
\subsection{Proof of correctness}
Let $G = (V, E)$ be a directed acyclic graph. $\level(v) \colon V \to \mathbb{N}$ is defined as follows
\begin{equation*}
 \level(v_i) = \left\{
      \begin{array}{cl}
       0  & \text{if } \deg^-(v_i) = 0,\\
       1 + \displaystyle \max_{v_j \in \pred(v_i)} \level(v_j) & \text{if } \deg^-(v_i) > 0
      \end{array} \right.
\end{equation*}
It is easy to see that $\level$ function is well defined.
\begin{figure}
\centering
\includegraphics{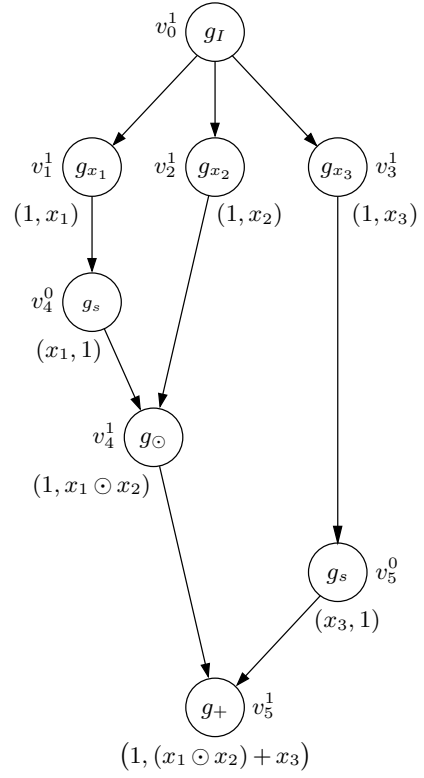}
\caption{Data flow graph corresponding to MCV instance in Figure \ref{fig_MCVP}}
\label{fig_MFPP}
\end{figure}
\begin{lemma} \label{main_lemma}
Let $\alpha = (C, \nu)$ be an instance of MCV with $C = (G, I, v_\theta, \tau)$ and $G = (V, E)$. Let $\alpha' = (G', D, M, v^1_\theta, (1, 1))$ be the instance of MFP as constructed in Section \ref{Subsec_MFP_redct}. Then $\overline{f}$ has a unique fixed point. For all $v_i \in V$, $\mfp(v^1_i) = (1, \nu(v_i))$ and for all $v_i \in V \setminus I$, if $v^1_j$ is the predecessor of $v^0_i$ in $G'$, then $\mfp(v^0_i) = (\nu(v_j), 1)$.
\end{lemma}
\begin{proof}
 $\mfp(\overline{f})$ exists by Theorem \ref{knaster_tarski_thrm}. Therefore, $\overline{f}$ has at least one fixed point. Let $\langle l_i \rangle$ be an arbitrary fixed point of $\overline{f}$. Let $\ell^0_j$ and $\ell^1_j$ denote the elements of $\langle l_i \rangle$ corresponding to vertices $v^0_j$ and $v^1_j$ respectively. We first prove that $\ell^k_i$ is uniquely defined for all vertices $v^k_i$ of $G'$.
 
 Since $\langle l_i \rangle$ is a fixed point of $\overline{f}$, so, $\overline{f}(\langle l_i \rangle) = \langle l_i \rangle$. By Equation \eqref{overline_f_def},  $\langle f^k_i (\displaystyle \bigwedge_{v^{k'}_j \in \pred(v^k_i)} \ell^{k'}_j ) \rangle = \langle l_i \rangle$. So,
 \begin{equation}
  \ell^k_i = f^k_{i} (\bigwedge_{v^{k'}_j \in \pred(v^k_i)} \ell^{k'}_j ) \qquad \forall v^k_i \in V' \label{Eq_2}
 \end{equation}
 $\ell^1_0 = g_I(\emptyset) = g_I(1) = 1$ is uniquely defined.
 Let $v_i$ be an arbitrary vertex of $V$. We prove the uniqueness of $\ell^k_i$ by induction on $\level(v_i)$.
 \begin{itemize}
  \item \emph{Base case:} $\level(v_i) = 0$ i.e.~$v_i \in I$. So, $f^1_i = g_{\nu(v_i)}$ and $g_{\nu(v_i)}((a_1, a_2)) = (1, \nu(v_i))$ $\forall (a_1, a_2) \in L$ by definition. From Equation \ref{Eq_2}, $\ell^1_i = f^1_i(\ell^1_0) = g_{\nu(v_i)}(\ell^1_0) = (1, \nu(v_i))$ is uniquely defined.
  \item \emph{Inductive step:}
  Let the theorem be true $\forall v_i \in V$ such that $\level(v_i) < m$. Let $\level(v_k) = m$. Let $v_i, v_j, i < j$ be predecessors of $v_k$ in $G$. By definition of $\level$ function, $\level(v_i) < \level(l_k) = m$ and $\level(v_j) < \level(l_k) = m$. By induction hypothesis, $\ell^1_i = (1, \nu(v_i))$ and $\ell^1_j = (1, \nu(v_j))$.
  
  From Equation \ref{Eq_2}, $\ell^0_k = g_{sw}(\ell^1_i) = g_{sw}((1, \nu(v_i))) = (\nu(v_i), 1)$ is uniquely defined.
  
  Let $\tau(v_k) = \odot$. From Equation \ref{Eq_2}, $\ell^1_k = f^1_k(\ell^0_k \bigwedge \ell^1_j) = g_{\odot}((\nu(v_i), 1) \bigwedge (1, \nu(v_j))) = g_{\odot}(\nu(v_i), \nu(v_j)) = (1, \nu(v_i) \odot \nu(v_j)) = (1, \nu(v_k))$ is uniquely defined.
    
  The case $\tau(v_k) = +$ is proved similarly.
\end{itemize}
Since $\ell^k_i$ is uniquely defined for all vertices $v^k_i$ of $G'$, $\langle l_i \rangle$ is unique and hence $\langle l_i \rangle$ is the maximum fixed point of $\overline{f}$. So, for all $v_i \in V$, $\mfp(v^1_i) = \ell^1_i = (1, \nu(v_i))$ and for all $v_i \in V \setminus I$, if $v^1_j$ is the predecessor of $v^0_i$ in $G'$, then $\mfp(v^0_i) = \ell^0_i = (\nu(v_j), 1)$.
\end{proof}
\begin{corollary}
 Let $\alpha = (C, \nu)$ be an instance of MCV with $C = (G, I, v_\theta, \tau)$ and $G = (V, E)$. Let $\alpha' = (G', D, M, v^1_\theta, (1, 1))$ be the instance of MFP as constructed in Section \ref{Subsec_MFP_redct}. $\nu(v_\theta) = 1 \iff \mfp(v^1_\theta) = (1, 1)$.
\end{corollary}
\begin{theorem}
 MFP is $\mathscr{P}$\textit{-complete}.
\end{theorem}
\begin{proof}
 A polynomial time algorithm for MFP is given in \cite{Kildall:1973:UAG:512927.512945}. MCV is shown $\mathscr{P}$\textit{-complete} in \cite{Goldschlager:1977:MPC:1008354.1008356}. It is easy to see that the above reduction is computable in log space. Hence MFP is $\mathscr{P}$\textit{-complete}.
\end{proof}
\section{An algorithm for GMR} \label{Section_GMR}
Algorithm \ref{algo_GMR} is an algorithm for deciding GMR.
\begin{algorithm}
\caption{Algorithm for GMR}
\begin{algorithmic}[1]
\Procedure{GMRA}{$G, A, L, v_{\theta\phi}, a_{\theta'}, l_{\phi'}$}
\State $temp \leftarrow 1$
\For{$i \leftarrow 1, n$}
\If{$v_{\theta' i}$ is reachable from $v_{\theta\phi}$}
\State $temp \leftarrow temp \bigwedge l_i$
\EndIf
\EndFor
\If{$temp = l_{\phi'}$}
\State \textbf{return} True
\Else 
\State \textbf{return} False
\EndIf
\EndProcedure
\end{algorithmic}
\label{algo_GMR}
\end{algorithm}

The observation below is a direct consequence of the above algorithm.
\begin{lemma}
 Let $\alpha = (G, A, L, v_{\theta\phi}, a_{\theta'}, l_{\phi'})$ be an instance of GMR. Let $R_{i} = \left\{ l_j \mid v_{ij} \text{ is reachable from } v_{\theta\phi} \right\}$. Then Algorithm \ref{algo_GMR} returns true $\iff \displaystyle \bigwedge_{l_i \in R_{\theta'} } l_i = l_{\phi'}$.
\end{lemma}
\begin{lemma}
 GMR is computable in non-deterministic log space.
\end{lemma}
\begin{proof}
 Variables $\mathit{temp}$ and $i$ take $\Theta(\log \lvert L \rvert)$ space. Since graph reachability takes up only non-deterministic log space \cite[Example 2.10 on p.~48]{Papadimitriou93}), Line 4 takes non-deterministic $\bigO(\log \lvert V \rvert) = \bigO(\log ( \lvert A \rvert \cdot \lvert L\rvert ))$ space. So, GMR is computable in non-deterministic log space.
\end{proof}
\section{MOP is NL-complete} \label{Section_MOP}
In this section, we give a log space reduction from MOP to GMR. Since GMR is non-deterministic log space computable, this implies that MOP can also be computed in non-deterministic log space.
\subsection[A reduction from MOP to GMR]{A reduction from MOP to GMR\footnote{This reduction is similar to the reduction in \cite{Reps:1995:PID:199448.199462}.}} \label{Subsec_MOP_redct}
Given an instance $\alpha = (G, D, M, v_\theta, l_\phi)$ of MOP with $G = (V, E, v_s, v_t)$, $D = (L, \bigwedge, F)$ and $L = \left\{0 = l_1, l_2, \ldots l_m = 1 \right\}$. Construct an instance of GMR $\alpha' = (G', A, L, v^0_{sm}, v^1_\theta, l_\phi)$ as follows
\begin{itemize}
 \item $A = \left\{ v^0_i \mid v_i \in V(G) \right\} \bigcup \left\{ v^1_i \mid v_i \in V(G) \right\}$
 \item $G' = (V', E')$ where
 \begin{itemize}
  \item $V' = \left\{ v^k_{ij} \mid v^k_i \in A, l_j \in L \right\}$
  \item $E'$ is defined as follows
  \begin{itemize}
   \item For each vertex $v_i \in V, l_j \in L$, if $f_i(l_j) = l_k$, add the edge $(v^0_{ij}, v^1_{ik})$ to $E'$.
   \item For each edge $(v_i, v_j) \in E, l_k \in L$, add the edge $(v^1_{ik}, v^0_{jk})$ to $E'$.
  \end{itemize}
 \end{itemize}
\end{itemize}
\begin{example}
Figure \ref{fig_MOPP} shows a Data Flow Graph and a lattice. Let the set $F$ of monotone functions be defined as follows: $F = \left\{ f_1, f_2, f_3, f_4, f_5, f_6 \right\}$ where
\begin{itemize}
 \item $f_1 = \left\{ (l_1, l_1), (l_2, l_4), (l_3, l_4), (l_4, l_3), (l_5, l_5) \right\}$
 \item $f_2 = \left\{ (l_1, l_1), (l_2, l_3), (l_3, l_3), (l_4, l_5), (l_5, l_5) \right\}$
 \item $f_3 = \left\{ (l_1, l_2), (l_2, l_3), (l_3, l_3), (l_4, l_2), (l_5, l_3) \right\}$
 \item $f_4 = \left\{ (l_1, l_1), (l_2, l_3), (l_3, l_5), (l_4, l_4), (l_5, l_5) \right\}$
 \item $f_5 = \left\{ (l_1, l_2), (l_2, l_2), (l_3, l_5), (l_4, l_3), (l_5, l_5) \right\}$
 \item $f_6 = \left\{ (l_1, l_1), (l_2, l_4), (l_3, l_5), (l_4, l_4), (l_5, l_5) \right\}$
\end{itemize}
Figure \ref{fig_GMR} shows the corresponding product graph.
\end{example}
\subsection{Proof of correctness}
\begin{lemma}
 Let $\alpha = (G, D, M, v_\theta, l_\phi)$ be an instance of MOP with $G = (V, E, v_s, v_t)$ and $D = (L, \bigwedge, F)$. Let $\alpha' = (G', A, L, v^0_{sm}, v^1_\theta, l_\phi)$ be an instance of GMR as constructed in Section \ref{Subsec_MOP_redct}. Then for all $v_i \in V$, there exists a $v_s$-$v_i$ path $p$ in G such that $f_p(1) = l_j \iff v^1_{ij}$ is reachable from $v^0_{sm}$ in $G'$.
\end{lemma}
\begin{proof}
 \begin{itemize}
  \item \emph{If part:} Let $v_i$ be an arbitrary vertex in $V$. Let $p$ be a $v_s$-$v_i$ path in $G$ and let $f_p(1) = l_j$ for some $l_j \in L$. We prove the if part by induction on $\len(p)$.
  
  \emph{Base case:} Let $\len(p) = 0$ i.e.~$v_i = v_s$. Then $l_j = f_{p}(1) = f_s(1)$. Then $(v^0_{sm}, v^1_{sj}) \in E'$. So, $v^1_{sj}$ is reachable from $v^0_{sm}$ in $G'$.
  
  \emph{Inductive step: } Let $\len(p) = k$ and let the if part be true for all paths from $v_s$ in $G$ with length less than $k$. Let $p = v_s, \ldots, v_{i'}, v_i$ for some $v_{i'} \in V$. Let $p'$ be the path $p$ with $v_i$ excluded i.e.~$p = p' \cdot v_i$ where $\cdot$ is the path concatenation operator. Therefore, $p'$ is a path from $v_s$ with length $k - 1$. Let $f_p(1) = l_{j'}$ and $f_i(l_{j'}) = l_j$ for some $l_{j'} \in L$. By induction hypothesis, $v^1_{i'j'}$ is reachable from $v^0_{sm}$ in $G'$. By construction of $E'$, $(v^1_{i'j'}, v^0_{ij'}), (v^0_{ij'}, v^1_{ij}) \in E'$. So, $v^1_{ij}$ is reachable from $v^0_{sm}$ in $G'$.
  
  So, the if part is true for all $v_s$-$v_i$ paths in $G$.
  \item \emph{Only if part:} Let $v^1_{ij}$ is reachable from $v^0_{sm}$ in $G'$. Let $p$ be a path in $G'$ from $v^0_{sm}$ to $v^1_{ij}$. It is easy to show that length of $p$ is odd. So, we prove the only if part by induction on $\len(p)$ where $\len(p)$ is odd.
  
  \emph{Base case:} Let $\len(p) = 1$ i.e.~$(v^0_{sm}, v^1_{ij}) \in E'$. By construction of $E'$, $v_i = v_s$ and $f_s(1) = l_j$. So, there exists a trivial $v_s$-$v_i$ path $q$ in G, the path having only one node $v_s$, such that $f_q(1) = l_j$.
  
  \emph{Inductive step:} Let $\len(p) = 2k + 1$ for some integer $k$ and let the only if part be true for all odd length paths from $v^1_{sm}$ in $G'$ with length less than $2k + 1$. Let $p = v^1_{sm}, \ldots, v^1_{i'j'}, v^0_{ij'}, v^1_{ij}$ for some $v_{i'} \in V, l_{j'} \in L$. Let $p'$ be the path $p$ with $v^0_{ij'}$ and $v^1_{ij}$ excluded i.e.~$p = p' \cdot v^0_{ij'} \cdot v^1_{ij}$ where $\cdot$ is the path concatenation operator. By induction hypothesis, there exists a $v_s$-$v_{i'}$ path $q$ in $G$ such that $f_q(1) = l_{j'}$. By construction of $E'$, $(v_{i'}, v_i) \in E$ and $f_i(l_{j'}) = l_j$. So, there exists a $v_s$-$v_i$ path $q' = q \cdot v_i$ in $G$ such that $f_{q'}(1) = f_i \circ f_q(1) = f_i(f_q(1)) = f_i(l_{j'}) = l_j.$
  
  So, the only if part is true.
 \end{itemize}
 So, the theorem is true.
\end{proof}
\begin{figure}
\centering
\includegraphics{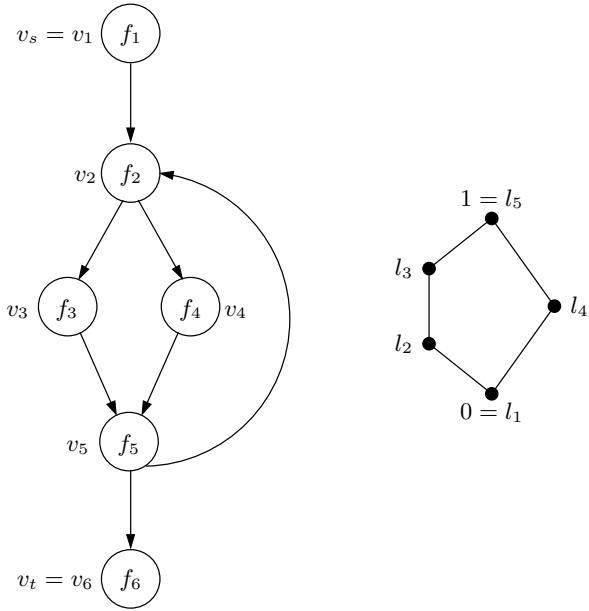}
\caption{A Data Flow Graph and a finite lattice}
\label{fig_MOPP}
\end{figure}
\begin{corollary}
 Let $\alpha = (G, D, M, v_\theta, l_\phi)$ be an instance of MOP with $G = (V, E, v_s, v_t)$ and $D = (L, \bigwedge, F)$. Let $\alpha' = (G', A, L, v^0_{sm}, v^1_\theta, l_\phi)$ be an instance of GMR as constructed in Section \ref{Subsec_MOP_redct}. Let $R_{\theta} = \left\{ l_i \mid v^1_{\theta i} \text{ is reachable from } v^0_{sm} \right\}$. Then $\mop(v_\theta) = l_\phi \iff \displaystyle \bigwedge_{l_i \in R_{\theta} } l_i = l_{\phi}$.
\end{corollary}
\begin{figure}
\centering
\includegraphics{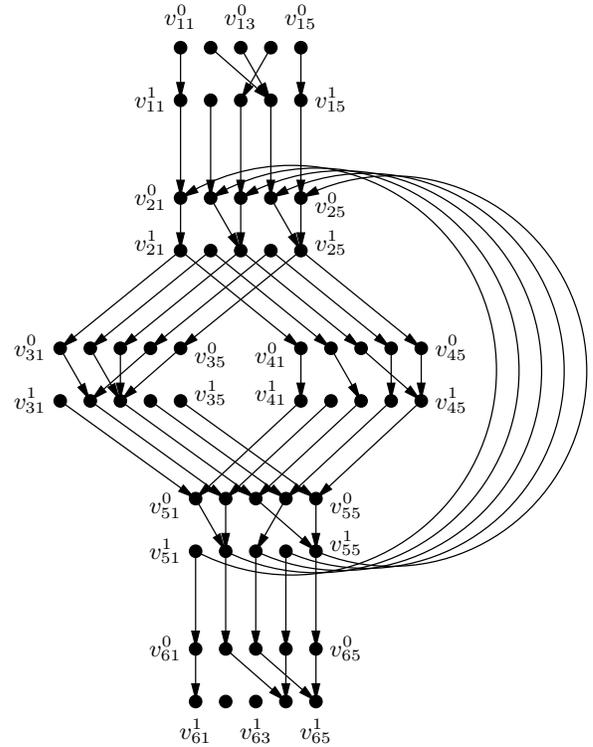}
\caption{The product graph corresponding to MOP instance in Figure \ref{fig_MOPP}}
\label{fig_GMR}
\end{figure}
\begin{theorem}
 $MOP$ is $\mathscr{NL}$\textit{-complete}.
\end{theorem}
\begin{proof}
 Algorithm \ref{algo_GMR} is $\mathscr{NL}$ computable. It is easy to see that the above reduction is log space computable. So, MOP is in $\mathscr{NL}$. Since \emph{Graph Reachability} (GR) problem is an instance of MOP and GR is $\mathscr{NL}$\textit{-complete}, it follows that MOP is also $\mathscr{NL}$\textit{-complete}.
\end{proof}
It is easy to see that the above proofs do not use the monotonicity of the data flow framework. So, the $\mathscr{NL}$\textit{-completeness} result holds even if the data flow framework is not monotone.
\newpage
\section{Acknowledgment}
We would like to thank Dr.~Vineeth Paleri and Ms.~Rekha R.~Pai for introducing us to the problem and for helpful discussions. We would also like to thank Dr.~Priya Chandran for reviewing our work.
%

\appendix
\section{Lattice representation}
\begin{algorithm}
\caption{Converts a lattice given as a covering relation to a lattice as a poset}
\begin{algorithmic}[1]
\Procedure{CovRel-to-Poset}{$L, \prec$}
\For{$i \leftarrow 1, n$}
 \For{$j \leftarrow 1, n$}
 \State Set $l_i \leqslant l_j$ to value False
 \EndFor
\EndFor
\For{$i \leftarrow 1, n$}
 \For{$j \leftarrow 1, n$}
  \If{$l_j$ is reachable from $l_i$}
   \State Set $l_i \leqslant l_j$ to value True
  \EndIf
 \EndFor
\EndFor
\EndProcedure
\end{algorithmic}
\label{algo_CovRel-to-Poset}
\end{algorithm}
\begin{algorithm}
\caption{Converts a lattice given as a poset to a lattice as an algebraic structure}
\begin{algorithmic}[1]
\Procedure{Poset-to-AlgStr}{$L, \leqslant$}
\For{$i \leftarrow 1, n$}
 \For{$j \leftarrow 1, n$}
  \State Set $l_i \bigvee l_j$ to value $1$
  \State Set $l_i \bigwedge l_j$ to value $0$
 \EndFor
\EndFor
\For{$i \leftarrow 1, n$}
 \For{$j \leftarrow 1, n$}
  \For{$k \leftarrow 1, n$}
   \If{$l_i \leqslant l_k$ \textbf{and} $l_j \leqslant l_k$}
    \If{$l_k \leqslant l_i \bigvee l_j$}
     \State Set $l_i \bigvee l_j$ to value $l_k$
    \EndIf
   \ElsIf{$l_k \leqslant l_i$ \textbf{and} $l_k \leqslant l_j$}
    \If{$l_i \bigwedge l_j \leqslant l_k$}
     \State Set $l_i \bigwedge l_j$ to value $l_k$
    \EndIf
   \EndIf
  \EndFor
 \EndFor
\EndFor
\EndProcedure
\end{algorithmic}
\label{algo_Poset-to-AlgStr}
\end{algorithm}
A lattice can be represented as a poset $(L, \leqslant)$ \cite[p.~33]{Davey2002Introduction}, as a covering relation for a poset $(L,\prec)$ \cite[p.~11]{Davey2002Introduction} or as an algebraic structure $(L, \bigvee, \bigwedge)$ \cite[p.~39]{Davey2002Introduction}. In this section, we give a non-deterministic log space algorithms to convert a lattice given as a poset or a covering relation to a lattice as an algebraic structure. This makes the completeness result of MOP w.r.t.~the class $\mathscr{NL}$ independent of the particular representation of the lattice.

A covering relation $(L, \prec)$ of a poset can be viewed as a graph where $L$ is the set of vertices and $\prec$ is the set of edges. So, reachability is defined as it is done for a graph.
\newpage
Algorithm \ref{algo_CovRel-to-Poset} (resp.~Algorithm \ref{algo_Poset-to-AlgStr}) converts a lattice given as a covering relation of a poset (resp.~a poset) to the lattice as a poset (resp.~an algebraic structure). The composition of the two algorithms converts a lattice given as a covering relation of a poset to the lattice as an algebraic structure.

\begin{lemma}
 Given a poset $(L, \leqslant)$ or a covering relation $(L, \prec)$) representation of a lattice, its algebraic structure representation $(L, \bigvee, \bigwedge)$ can be computed in non-deterministic log space.
\end{lemma}
\begin{proof}
  Line 9 of Algorithm \ref{algo_CovRel-to-Poset} takes non-deterministic log space since GR takes non-deterministic log space \cite[Example 2.10 on p.~48]{Papadimitriou93}. All other lines of the two algorithms take at most log space.
  So, Algorithm \ref{algo_CovRel-to-Poset} takes non-deterministic log space while Algorithm \ref{algo_Poset-to-AlgStr} takes log space. The composition of the two algorithms takes non-deterministic log space \cite[Theorem 8.23 on p.~324]{sipser2005}. So, the conversions can be done in non-deterministic log space.
\end{proof}
It may be noted that the lattice in Section \ref{Section_MFP} is of constant size. So, it can be converted to a poset or a covering relation in constant time. So, the completeness result of MFP w.r.t.~the class $\mathscr{P}$ is also independent of the particular representation of the lattice.
%
%
\end{document}